%% file: main.tex
\documentclass[sigplan, screen]{acmart}

\renewcommand\footnotetextcopyrightpermission[1]{} % removes footnote with conference information in first column
\usepackage{tikz}
\usepackage{amsmath}
\usepackage{url}
\usepackage{amsthm}
\usepackage{subfig}
\usepackage{graphicx}
\usepackage{soul}
\usepackage[linesnumbered, ruled,vlined]{algorithm2e}
\usepackage{amsmath}
\usepackage{enumitem}
\usepackage{xcolor}
\usepackage{ctable}
\usepackage{multirow}
\usepackage{listings}

\newtheorem{theorem}{Theorem}

\theoremstyle{remark}
\newtheorem{definition}{Definition}

\definecolor{red}{rgb}{1,0,0}
\definecolor{blue}{rgb}{0,0,1}

\definecolor{mygreen}{rgb}{0.23,0.33,0.16}

\SetCommentSty{mycommfont}

\begin{document}
\title{Efficient Mining of Frequent Subgraphs with Two-Vertex Exploration}

%%
%% The "author" command and its associated commands are used to define the authors and their affiliations.
\author{Peng Jiang}
\affiliation{%
  \institution{The University of Iowa}
  %\streetaddress{P.O. Box 1212}
 % \city{Dublin}
 % \state{Ireland}
 % \postcode{43017-6221}
}
\email{peng-jiang@uiowa.edu}

\author{Rujia Wang}
%\orcid{0000-0002-1825-0097}
\affiliation{%
  \institution{Illinois Institute of Technology}
  %\streetaddress{1 Th{\o}rv{\"a}ld Circle}
 % \city{Hekla}
 % \country{Iceland}
}
\email{rwang67@iit.edu}

\author{Bo Wu}
%\orcid{0000-0001-5109-3700}
\affiliation{%
  \institution{Colorado School of Mines}
  %\city{Rocquencourt}
 % \country{France}
}
\email{bwu@mines.edu}

%%
%% The abstract is a short summary of the work to be presented in the
%% article.
\begin{abstract}
Frequent Subgraph Mining (FSM) is the key task in many graph mining and machine learning applications. 
Numerous systems have been proposed for FSM in the past decade. 
Although these systems show good performance for small patterns (with no more than four vertices), we found that they have difficulty in mining larger patterns. 
In this work, we propose a novel {\em two-vertex exploration} strategy to accelerate the mining process.   
Compared with the single-vertex exploration adopted by previous systems, our two-vertex exploration avoids the large memory consumption issue and significantly reduces the memory access overhead.  
We further enhance the performance through an {\em index-based quick pattern} technique that reduces the overhead of isomorphism checks, and a {\em subgraph sampling} technique that mitigates the issue of subgraph explosion.  
The experimental results show that our system achieves significant speedups against the state-of-the-art graph pattern mining systems and supports larger pattern mining tasks that none of the existing systems can handle. 
\end{abstract}

\maketitle

\pagestyle{plain} % removes running headers

\input{text/intro.tex}

\input{text/background.tex}
\input{text/overview.tex}

\input{text/multi_join.tex}

\input{text/approx_join.tex}

\input{text/experiments}

\input{text/relatedwork}
\input{text/conclusions}

\bibliographystyle{ACM-Reference-Format}
\bibliography{reference.bib}

\end{document}

%% file: text/intro.tex
\vspace{-1em}
\section{Introduction}
%\vspace{-0.3em}
Frequent Subgraph Mining (FSM) is an important operation on graphs and is widely used in various application domains, including bioinformatics~\cite{Milo824,  Vazquez17940}, computer vision~\cite{10.1145/2324796.2324831}, and social network analysis~\cite{10.1145/2488388.2488502}.  
The task is to discover frequently occurring subgraph patterns from an input graph. 
Different from graph pattern matching problems where a query pattern is given, FSM needs to find the important patterns based on a support measure and thus has a much larger exploration space. 
%For example, motif counting (MC)~\cite{Alon2008BiomolecularNM} counts the embeddings of each of the connected patterns of a particular size without knowing the exact topology of the patterns; frequent subgraph mining (FSM)~\cite{RAMRAJ2015197} needs to discover all patterns of a particular size whose frequency exceeds a threshold. 
%Compared with GPMA problems, GPMI applications have a broader search space to reveal unknown information in a graph.

Since the patterns of interest are unknown, most systems for FSM take an {\em explore-aggregate-filter} approach~\cite{10.1145/2815400.2815410, 10.1145/3299869.3319875, 222571, 10.14778/3389133.3389137}. 
The principle is to explore all the subgraphs, aggregate the subgraphs according to their patterns, and filter out the subgraphs that are redundant or are not of interest. 
The exploration happens in a vertex-by-vertex manner where smaller subgraphs are iteratively extended based on the connections in the graph. 
There are mainly two ways for exploration: breadth-first and depth-first. 
Starting from all vertices in the graph, breadth-first exploration stores all subgraphs of size $l$ and extends them with one more vertex to find subgraphs of size $l+1$. 
The main problem of breadth-first exploration is that the intermediate data can easily exceeds the memory capacity as the subgraph size grows.   
With depth-first exploration, a subgraph of size $l$ is immediately extended to a subgraph of size $l+1$ without seeing other subgraphs of size $l$. 
It needs not save the intermediate subgraphs and thus can explore larger patterns. 
However, depth-first exploration cannot exploit the {\em anti-monotone} property to prune the search space~\cite{10.1145/3299869.3319875}, resulting in a lot of unnecessary computation.

Some recent graph mining systems take a {\em pattern-based} approach~\cite{10.1145/3341301.3359633, 10.1145/3342195.3387548}.  
The idea is to enumerate the (unlabeled) subgraph patterns and then perform pattern matching on the graph. 
Because the pre-generated patterns guide the exploration, these systems need not store any intermediate data, and the aggregation overhead can be reduced as the topology of the subgraphs is given. 
However, this approach only works well for small patterns because when the pattern is larger (more than 6), listing all subgraph patterns itself becomes a hard problem~\cite{MCKAY201494, nauty, ng}. 
It is also difficult for the pattern-based systems to exploit the {\em anti-monotone} property to prune the search space. 
Peregrine~\cite{10.1145/3342195.3387548} maintains a list of frequent patterns, extend the patterns with one vertex or edge, and then re-match the extended patterns on the graph. It prunes the search space without storing the intermediate subgraphs, but the re-matching incurs a lot of redundant computation. 
These issues have impeded the existing graph mining systems from supporting FSM for large patterns. 
In fact, most of the prior work only reports experimental results for FSM with no more than 4 vertices.

 %The scalability of current systems is limited, while mining larger patterns is desirable as it can provide more structural information of the graphs.
%while many real-world applications desire to mine larger patterns. 
%\textcolor{red}{We need to add some reference here}
%\rw{Tried to find some references but cannot find exact argument we're looking at.}
%For example, a user might be interested in finding the relations among a group of people (greater than 5) in a social network. 

% The limitations of current GPMI systems expose two challenges in supporting graph mining with large patterns. 
% The fundamental challenge comes from the explosion of subgraphs as the size increases. 
% For example, 
% even in a median-size graph, MiCo~\cite{10.14778/2732286.2732289}, which has $9\times 10^4$ vertices and $10^6$ edges, there are more than $2\times 10^{12}$ size-5 subgraphs, and storing them needs more than 18TB disk space.  
% When the pattern size increases to $7$, even counting the number of subgraphs becomes intractable.  
% The second challenge is that testing isomorphism for large patterns is  expensive as the problem is NP-hard~\cite{10.14778/1453856.1453899, 4568073}.  
% Since listing all the large patterns is impossible (especially when labels are considered), we have to perform isomorphism test to group the subgraphs into different patterns. 
% In our experiments, the isomorphism test can take more than 90\% of the total execution time of a size-9 FSM task with the existing techniques. 

To enable large pattern mining, we propose a novel \textit{two-vertex exploration} method in this work. 
Our key observation is that vertex-by-vertex exploration is not necessary for pattern mining. 
Instead, we can perform two-vertex exploration that joins size-($n-2$) subgraphs with size-$3$ subgraphs on a common vertex to obtain subgraphs of size-$n$. 
%If we want to obtain all the size-$n$ subgraphs, we can join all the size-($n-2$) subgraphs with all the size-$3$ subgraphs on a common vertex (i.e., two-vertex exploration), and in some cases, we can even do three-step exploration. 
The new exploration method significantly accelerates the exploration process and reduces the memory access overhead in the join operation. It also allows us the exploit the anti-monotone property to prune the exploration space without storing the intermediate subgraphs or re-matching the patterns.  
%It also allows us to exploit the highly-optimized pattern matching algorithms 
%(e.g., AutoMine~\cite{10.1145/3341301.3359633}, DAF~\cite{10.1145/3299869.3319880}) to generate small subgraphs. 
%Also, the corresponding multi-way join is much faster than simple vertex-by-vertex join with a lower worst-case complexity~\cite{10.14778/3342263.3342643, 10.1145/3180143}.  
%Compared with the pure pattern-based approach, we discover the large patterns by joining the smaller subgraphs, so there is no need to list all the large patterns. 
%Based on the multi-vertex exploration computation model, we propose a \textit{match-and-join} programming interface. 
%With our interface, the exploration of large subgraphs can be easily expressed as a {\em multi-way join} of small subgraphs. 
%By exposing the match and the join functions, our interface is flexible for expressing both the explore-aggregate-filter mining approach and the pattern-based approach. 

To further accelerate the mining process, we propose two new techniques to overcome the performance bottlenecks. 
One performance bottleneck is due to the expensive isomorphism checks in the aggregation step. 
To aggregate the subgraphs based on their patterns, we need to generate a {\em canonical form} for each subgraph such that the subgraphs with the same canonical form are isomorphic. 
Unfortunately, the best known algorithms for generating such canonical forms have exponential complexity~\cite{10.14778/1453856.1453899, 4568073, 1184038}.  
Therefore, we want to perform isomorphism check for as few subgraphs as possible. 
Previous work has employed a {\em quick pattern} technique to reduce the number of isomorphism checks~\cite{10.1145/2815400.2815410, 222571}. 
The main is to first group the subgraphs based on an easily computed pattern (e.g., a list of all edges). 
%The subgraph is first processed to generate a quick pattern with linear time complexity (e.g., by listing its vertices/edges in a particular order). 
Since subgraphs in the same group must be isomorphic,   only one isomorphism check is needed for each group. 
We improve on this idea by proposing an {\em index-based quick pattern} technique. 
It assigns an index to each pattern and uses the indices to  compute a quick pattern for the joined subgraph. 
Compared with the quick pattern technique used in prior work, 
our quick pattern encodes the information of sub-patterns and achieves more accurate grouping of the subgraphs, leading to a significant reduction of isomorphism checks.

Another more fundamental challenge of mining large patterns on graphs is due to the exponential growth of the exploration space. 
For example, in a median size graph, MiCo~\cite{10.14778/2732286.2732289}, which
has $9\times 10^4$ vertices and $10^6$
edges, there are more than $10^{12}$ size-5 subgraphs. When the pattern size increases to 7, the estimated
number of subgraphs is in the order of $10^{17}$ for which exhaustive enumeration becomes infeasible. 
To mitigate this issue, we propose a {\em subgraph sampling} technique. 
The idea is that we sample a small subset of size-3 subgraphs for exploring larger subgraphs during the joining and/or the matching phase. 
Since the subgraphs of frequent patterns are more likely to be sampled, we are able to discover frequent patterns with only a small number of sampled subgraphs. 
Compared with previous works that apply edge or neighbor sampling to FSM~\cite{222637, 8411069}, we can discover more frequent patterns with the same or less computation. This is because subgraph samples preserve more structural information of the graph than edge samples.

We perform extensive evaluation of our system and compare with three state-of-the-art graph mining systems: AutoMine~\cite{10.1145/3341301.3359633}, Peregrine~\cite{10.1145/3342195.3387548},  and Pangolin~\cite{10.14778/3389133.3389137}.   
The results show that without using sampling our system achieves 1.8x to 8.4x speedups on tasks for which the compared systems can return. 
By using sampling, our system can discover larger patterns that none of the existing systems can handle. 

%In the evaluation, we show that {\systemName} outperforms the state-of-the-art GPMI systems by 2x to 160x for pattern mining tasks of different sizes on median size graphs and can return results in a reasonable amount of time on large graphs that none of the existing systems can handle. % \rw{check the numbers}

%% file: text/background.tex
\section{Background}
%\vspace{-0.3em}
This section introduces the graph related concepts that are important to our discussion and formally defines the frequent subgraph mining problem. 

% \subsection{GPMI Basics}
\label{sec:appl}
% \vspace{-0.5em}
\subsection{Graph Basics}
A {\em graph} $G$ is defined as $G = (V, E, L)$ consisting of a set of vertices $V$, a set of edges $E$ and a labeling function $L$ that assigns labels to the vertices and edges. 
A graph $G' = (V', E', L')$ is a {\em subgraph} of graph $G = (V, E, L)$ if $V'\subseteq V$, $E'\subseteq E$ and $L'(v)=L(v), \forall v\in V'$. 
A subgraph $G' = (V', E', L')$ is {\em vertex-induced} if all the edges in $E$ that connect the vertices in $V'$ are included $E'$. 
A subgraph is {\em edge-induced} if it is connected and is not vertex-induced.

\begin{definition}[Isomorphism]
  Two graphs $G_a=(V_a, E_a, L_a)$ and $G_b=(V_b, E_b, L_b)$ are isomorphic if there is a bijective function $f: V_a\Rightarrow V_b$ such that $(v_i, v_j)\in E_a$ if and only if $(f(v_i), f(v_j))\in E_b$.
\end{definition}

We say two (sub)graphs have the same {\em pattern} if they are isomorphic. The pattern is a template for the isomorphic subgraphs, and a subgraph is an instance (also called {\em embedding}) of its pattern. 
To determine the pattern of a subgraph, a canonical form for each subgraph can be computed, and the subgraphs with the same canonical form are isomorphic. 
There are various tools and algorithms available for graph isomorphism check~\cite{MCKAY201494, JunttilaKaski:ALENEX2007, 1184038}. All of these algorithms have exponential complexity.  
We use bliss~\cite{JunttilaKaski:ALENEX2007} for isomorphism check in our system as it is fast in practice and is widely used in graph mining systems~\cite{10.1145/2815400.2815410, 222571, 10.1145/3342195.3387548}. 
A related concept is {\em automorphism check} which checks if two subgraphs are identical, even though they might have different orderings of vertices and edges.

\subsection{Frequent Subgraph Mining} 
The task of Frequent Subgraph Mining (FSM) is to obtain all frequent subgraph patterns from a labeled input graph. A pattern is considered frequent if it has a {\em support} above a threshold. 
  While the definition of the support measure can vary across applications, the support usually needs to satisfy an {\em anti-monotone} property, i.e., the support of a pattern should be no greater than the support of its sub-patterns~\cite{10.1145/3035918.3035936}.

  \begin{definition}[MNI Support]
  Given a pattern $P=(V_p, E_p, L_p)$ and an input graph $G=(V,E,L)$, if $P$ has $m$ embeddings $\{f_1, f_2, \ldots, f_m\}$ in $G$, the minimum image based (MNI) support of $P$ in $G$ is defined as
  %\vspace{-0.3em}
  \[
  \sigma_{MNI}(P,G) = \min_{v\in V_p} |\{f_i(v): i=1,2,\ldots,m\}|. 
  \]
\end{definition}
\vspace{-0em}
Other support measures include maximum independent set based (MIS) support, minimum instance based (MI) support, and maximum vertex cover based (MVC) support. 
All these support measures are anti-monotone. MNI support is the most  commonly used one because it has linear computation complexity while achieving a good accuracy in measuring the `frequency' of patterns in a graph. 
The readers are refered to~\cite{10.1145/3035918.3035936} for detailed descriptions and computation complexity of different support measures. 
We adopt the MNI support for our experiments, although our proposed techniques are applicable to any support measure with the anti-monotone property. 
%Some existing GPMI systems support mining both vertex-induced and edge-induced subgraphs~\cite{10.1145/3299869.3319875, 222571, 10.1145/3342195.3387548}, while some others only support vertex-induced mining~\cite{10.1145/3341301.3359633, 222637}. 
%We focus on vertex-induced mining in this work, but the techniques can be extended to support edge-induced mining. 

With a support measure $\sigma$, the frequent subgraph mining problem is defined as finding all patterns $\{P_i=(V_i, E_i, L_i)\}$ in a graph $G$ such that $|V_i|=s$ and $\sigma(P_i, G)\geq t$ where $s$ is the given pattern size and $t$ is the given support threshold. 
The support can be calculated with either vertex-induced subgraphs or edge-induced subgraphs. 
Our proposed techniques work for both cases. We use edge-induced subgraphs for  experiments, as it is the common setting in prior work~\cite{222571, 10.1145/3299869.3319875, 10.1145/3342195.3387548}.

%% file: text/overview.tex
\section{Illustration of the Idea}
\label{sec:overview}
%\vspace{-0.3em}
\begin{figure}[t]
  \centering
 \subfloat[An example graph]{
   \includegraphics[scale=0.36]{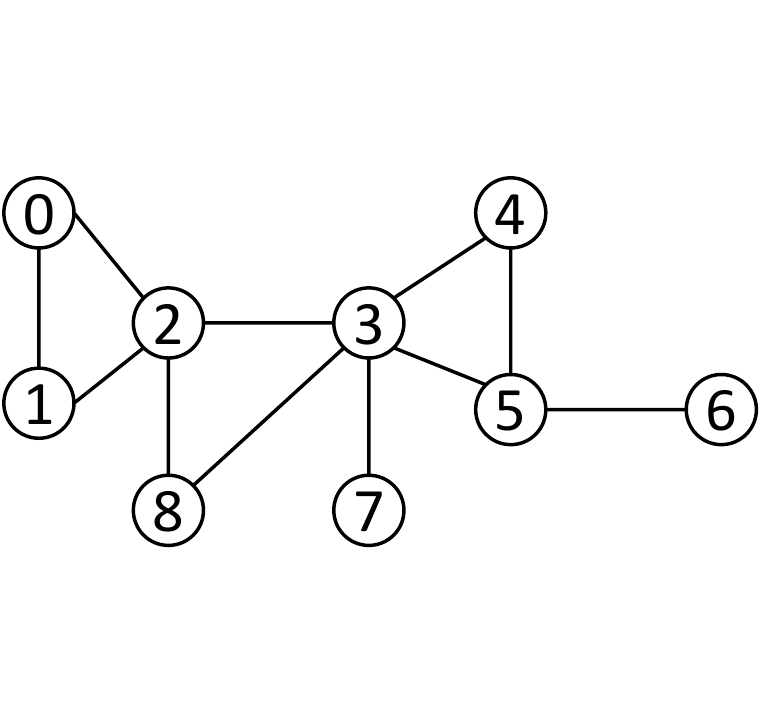}
   \label{fig:reg_graph}
 } \hfil
 \subfloat[Size-3 subgraphs]{
   \includegraphics[scale=0.39]{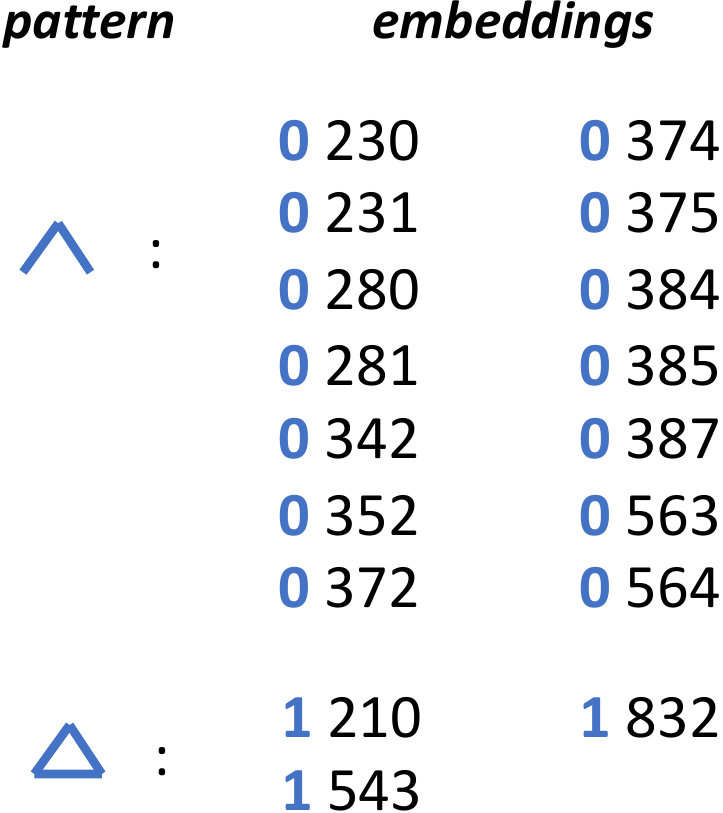}
   \label{fig:reg_size3}
 }\\ \vspace{-0.6em}
 \subfloat[Join the size-3 subgraphs on every column]{
   \includegraphics[scale=0.39]{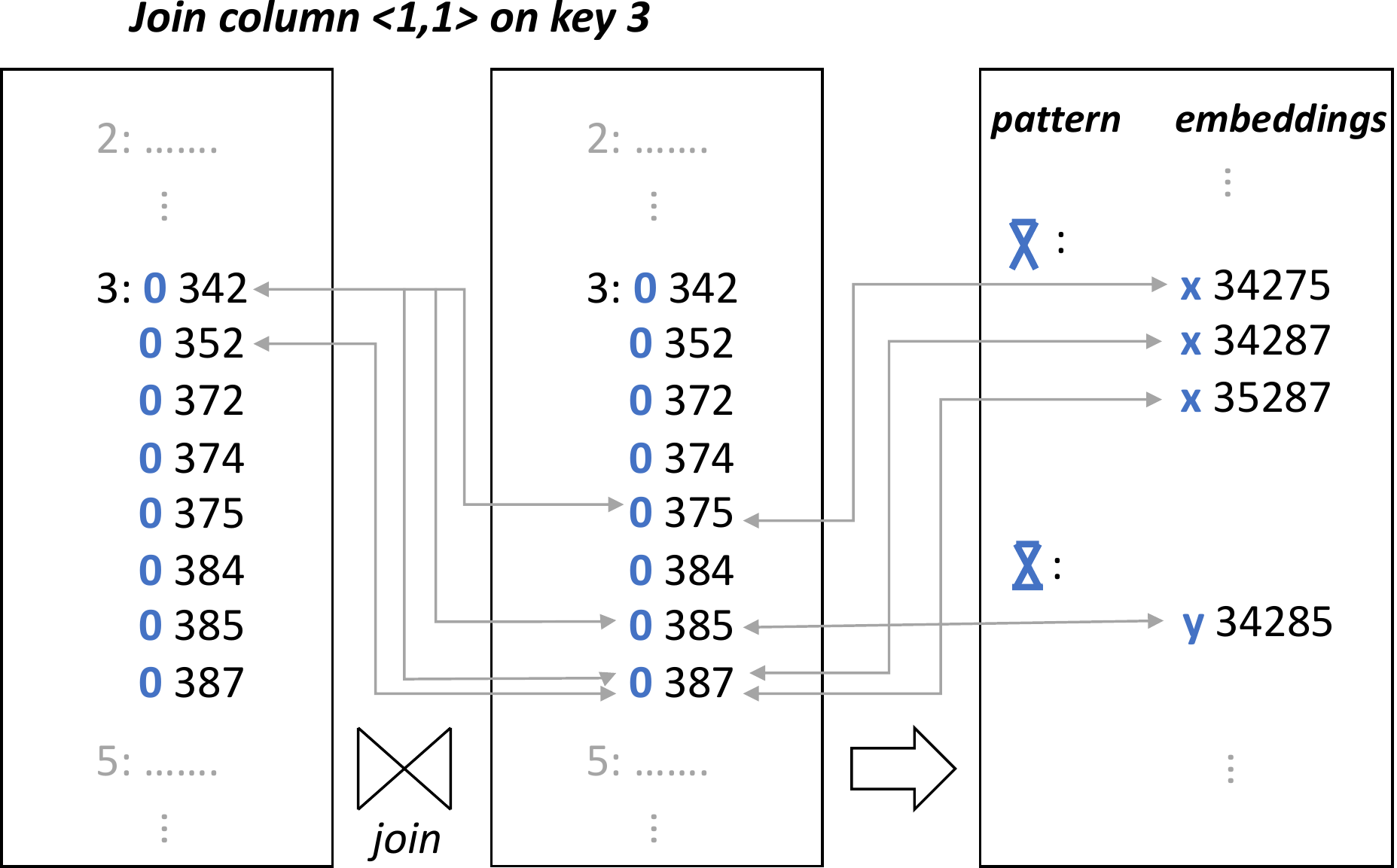}
   \label{fig:reg_join}
 }
     \vspace{-.5em}
 \caption{An example of finding size-5 subgraphs by joining size-3 subgraphs.}
\vspace{-1em}
 \label{fig:reg}
\end{figure}

Before getting into technical details, we describe our idea of
two-vertex exploration with an example. 
We use an unlabeled graph for simple illustration. 

Suppose our task is to discover size-5 patterns in an input graph (as shown in Figure~\ref{fig:reg_graph}). 
We can first find all size-3 subgraphs and join them on a common vertex to obtain size-5 subgraphs. 
In this example, we first apply a matching algorithm to obtain all the embeddings of size-3 patterns (i.e., wedge and triangle) as listed in Figure~\ref{fig:reg_size3}.  
Each pattern is assigned an index (0 for wedge and 1 for triangle in this example), and  the index is stored with each embedding during the pattern matching.  
%The pattern indices will be used for computing a quick pattern when we join two subgraphs. 

Next, we calculate the MNI support for each size-3 pattern and prune the patterns with support less than the threshold. 
In this example, the supports for both wedge and triangle is 3. 
Suppose we set the support threshold to 3. Neither of the patterns will be pruned. After we obtain the pruned size-3 subgraphs, we perform binary join on every pair of the columns (i.e., $\langle 1,1 \rangle$, $\langle 1,2 \rangle$,$\langle 1,3 \rangle$,$\langle 2,1 \rangle$,$\langle 2,2 \rangle$,$\langle 2,3 \rangle$,
$\langle 3,1 \rangle$,$\langle 3,2 \rangle$,$\langle 3,3 \rangle$) to explore size-5 subgraphs. 
%As we will explain in \S\ref{sec:validity}, the binary join of size-3 subgraphs ensures all size-5 subgraphs are explored.   
Figure~\ref{fig:reg_join} shows how we can obtain four size-5 subgraphs by joining column $\langle 1,1 \rangle$ on key $3$.  
Every pair of subgraphs with key $3$ are tested (i.e., $\langle$`342', `342'$\rangle$, $\langle$`342, `352'$\rangle$, ..., $\langle$`387', `385'$\rangle$, $\langle$`387', `387'$\rangle$). 
If two subgraphs have one and only one common vertex, they compose a valid size-5 subgraph. 
In this example, `342' and `375' make up a valid size-5 subgraph `34275'; `342' and `387' make up `34287'; `352' and `387'  make up `35287'; and `342' and `385' make up `34285'. These valid joins are marked with connected arrows in Figure~\ref{fig:reg_join}. 
We can see that, through the join operation, we grow the pattern size from 3 to 5 in one exploration step. 
We will show that such two-vertex exploration is exhaustive for subgraph exploration in \S\ref{sec:validity}.

One may notice that the result of joining `374' with `385' (`37485') is not included in Figure~\ref{fig:reg_join}. 
This is because our system performs an automorphism check when generating the join results to remove redundancy. We propose a \textit{smallest-vertex first} dissection method that ensures only the results that are obtained by joining the subgraph of the smallest spanning vertex indices are saved. 
In this case, the `37485' subgraph will be generated when we join the third column of `543' and the first column of `387'.
%The reason is that `543' is the subgraph that has smallest spanning vertex indices, and `387' is also a connected subgraph.
%This ensures that the join operation  produces only one result for each pattern. 
More details on the automorphism check and redundancy removal are explained in \S\ref{sec:redun_removal}.

The above procedure can be extended to explore larger subgraphs by joining multiple subgraph lists. 
For example, a 3-way join of two size-3 subgraphs and one size-2 subgraphs (i.e. edges) will explore all size-6 subgraphs. A 3-way join of size-3 subgraphs will explore all size-7 subgraphs. 
Given an input graph $G$, a pattern size $s$, and a support threshold $t$, the workflow of our frequent subgraph mining algorithm is summarized as follows:

\begin{enumerate}[label=, leftmargin=1\parindent, itemsep=1ex, topsep=1ex]
    \item Step1: Obtain all size-3 subgraphs by matching.
\item Step2: Calculate the support for each size-3 and size-2 pattern, and remove patterns with support smaller than $t$ along with their subgraphs. 
\item Step3: Perform multi-way join of size-3 subgraphs and/or edges to obtain subgraphs of size s: 
if $s==2n+1$, join $n$ size-3 subgraph lists;
if $s==2n$, join the edge list with $n-1$ size-3 subgraph lists. 
\item Step4: Calculate support for each size-$s$ patterns and remove patterns with support smaller than $t$. 
\end{enumerate}
\noindent
For Step1, any matching algorithm will work; we use AutoMine~\cite{10.1145/3341301.3359633} in our implementation. 
Step2 and Step4 are straightforward based on the definition of the support measure. 
Step3 is the most important step in the algorithm. We will detail this step in the next section.

%% file: text/multi_join.tex
\section{Subgraph Exploration Process}
%\vspace{-0.2em}

%\section{Subgraph Exploration through Joining Small Subgraphs}
\label{sec:explore}
%\hl{title}
 All the current graph mining systems based on the explore-aggregate-filter approach use single-vertex exploration because it ensures that all the size-$n$ subgraphs can be found by extending the size-($n-1$) subgraphs with an edge. 
 %In other words, the exploration step size in this case is 1. 
 We find that limiting the step size to 1 is not a must to find all patterns. 
 This section describes our two-vertex exploration idea and explains its advantage over single-vertex exploration. 

\subsection{Two-Vertex Exploration}
\label{sec:validity}
%\vspace{-0.5em}

We propose to explore the size-$n$  subgraphs by joining the size-($n-2$) subgraphs with the size-3 subgraphs (i.e., wedges and triangles). 
The completeness of this two-vertex exploration method is summarized as follows.

\begin{theorem}
\label{th:1}
All of the size-$n$ subgraphs can be discovered by joining the size-($n-2$) subgraphs with the size-$3$ subgraphs on a common vertex. 
\end{theorem}
\begin{proof}
Our goal is to show that any size-$n$ subgraph can be dissected into a connected size-($n-2$) subgraph and a connected size-$3$ subgraph on one vertex. 
Because we join all size-($n-2$) and size-$3$ subgraphs in all possible ways, if a dissection exists for a size-$n$ subgraph, it will be discovered by the join operation.
Suppose any size-$n$ subgraph can be dissected into a size-$(n-2)$ and a size-$3$ subgraph. 
There are only two way a size-($n+1$) subgraph can be constructed from a size-$n$ subgraph: 1) the new vertex is connected with the size-$(n-2)$ subgraph, and in this case, the size-($n+1$) subgraph can be dissected in the same way as the size-$n$ subgraph; 
2) if the new vertex is only connected with the size-$3$ subgraph, it is easy to verify that for either of the two cases (wedge or triangle), we can always pick three connected vertices as the new dissection. 
As the base case, all the six size-$4$ patterns can be dissected into a size-$3$ subgraph and an edge. 
The proof finishes by induction. 
\end{proof}

Note that multi-vertex exploration is not complete with more than two vertices. 
For example, a seven-vertex three-pronged star graph with two vertices in each prong cannot be obtained by joining any two size-4 subgraphs. 
Therefore, we cannot explore more than two vertices in each step.

Two-vertex exploration can be either vertex-induced or edges induced. 
For vertex-induced exploration, we add all the connecting edges between the two joining subgraphs to the resulting subgraph. 
For edge-induced exploration, we enumerate all possible combinations of the connecting edges between the joining subgraphs and generate a resulting subgraph for each combination.

%The benefit of growing more than one vertex in each step is that we can reduce the number of intermediate join operations. To verify this point, we count the number of join operations with single-vertex exploration and two-vertex exploration for size-5 and size-7 motif counting on two graphs: CiteSeer and MiCo~\cite{10.14778/2732286.2732289}. 
%The number of join operations is reduced by a factor of 1.73 for size-5 motif counting and by a factor of 5.24 for size-7 motif counting with two-vertex exploration. 
%On MiCo, the number of join operations is reduced by a factor of 2.25 for size-5 motif counting and 8.27 for size-7 motif counting.  

\subsection{Depth-First Multi-Way Join of Subgraphs}
\label{sec:multi-join}
%\vspace{-0.5em}

\begin{figure}[t]
    \centering
    \includegraphics[scale=0.39]{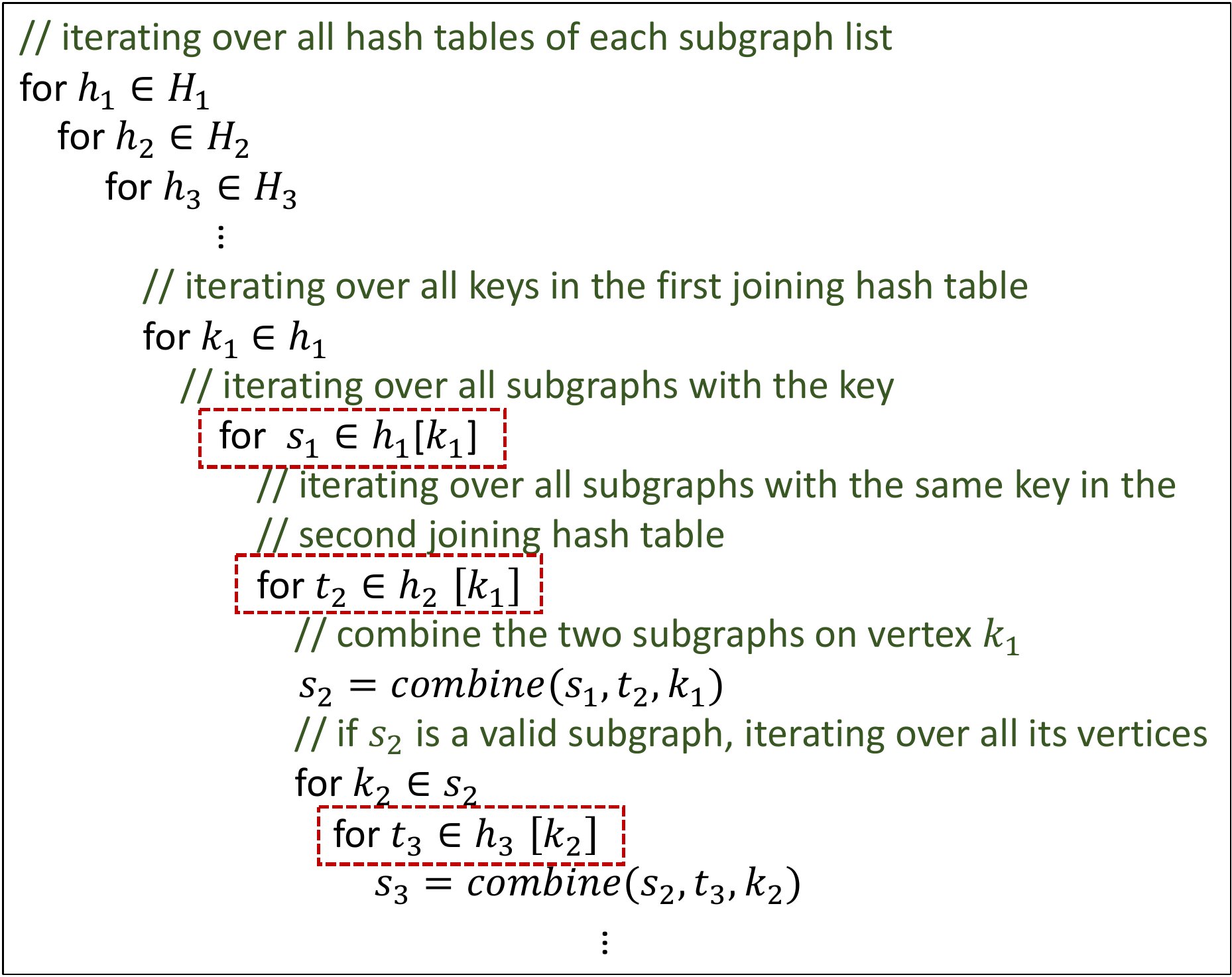}
  \vspace{-1em}
    \caption{Code of multi-way join.}
    \vspace{-.5em}
    \label{fig:join_code}
\end{figure}

%Multi-vertex exploration reduces the number of join steps; however, it still needs to store a large amount of intermediate data if we perform the join step-by-step. 
To avoid the large memory consumption, we implement the exploration process as a depth-first multi-way join. 
%For example, we know that all size-6 subgraphs can be obtained by joining size-3 and size-4 subgraphs, and all size-4 subgraphs can be obtained by joining size-3 subgraphs with size-2 subgraphs (edges). 
% To explore the size-6 subgraphs, we can perform a 3-way join of size-2, size-3, size-3 subgraphs. 
Suppose we want to join $t$ subgraph lists $SL_1, SL_2, ..., SL_t$ and the subgraph in $SL_i$ has $l_i$ vertices. 
For each subgraph list $SL_i$, we group the subgraphs by each of its $l_i$ columns, create $l_i$ hash tables and store the hash tables in $H_i$.  
For example, the size-$3$ subgraphs in Figure~\ref{fig:reg_join} are grouped by the vertex indices in the first column. 
Once the hash tables are created, the multi-way join operation is simply a nested loop that iterates over all possible combinations of subgraphs in different hash tables, as shown in Figure~\ref{fig:join_code}. 
We first enumerate all possible combinations of columns in different subgraph lists by iterating over all hash tables of each subgraph list. 
Then, we identify the matching keys $k_1$ in the first two hash tables and try to combine the subgraphs ($s_1$ and $t_2$) on the key. 
If the two subgraphs make up a valid larger subgraph $s_2$, we iterate over all the vertices of $s_2$ and look up each vertex $k_2$ in the third hash table. 
For every subgraph $t_3$ with key $k_2$ in the third hash table, we  combine $s_2$ with $t_3$ to obtain a larger subgraph. 
For joining more subgraph lists, the code simply repeats the for loop. 

Depth-first join is also used in Fractal~\cite{10.1145/3299869.3319875} for single-vertex exploration. 
The main issue is that it incurs a huge amount of redundant memory accesses. 
Our two-vertex exploration mitigates this issue as it requires fewer join steps to enumerate subgraphs of a certain size. 
To see this, let us consider the exploration of size-5 subgraphs. 
With single-vertex exploration, it requires a 4-way join of the edges in the graph. 
The first join operation of the edge list does not incur any redundant memory accesses as each neighbor list is accessed only once in the two hash tables. 
However, when we join the intermediate size-3 subgraphs with the edge list, we need to query the edge list for each intermediate subgraph. 
For non-consecutive size-3 subgraphs of the same key, 
each neighbor list will be accessed multiple times during the join process. 
The same problem exists when joining size-4 subgraphs with the edge list. 
In contrast, two-vertex exploration obtains size-5 subgraphs by performing a binary join of size-3 subgraphs which incurs no redundant memory accesses. 
Our experimental results also validate this point.

\subsection{Redundancy Removal through Smallest-Vertex-First Dissection}
\label{sec:redun_removal}
%\vspace{-0.5em}

\begin{algorithm}[t]
\small
\SetKwInOut{Input}{Input}
        \SetKwInOut{Output}{Output}
        \Input{$\;\;\;\text{subgraph } s; \text{subgraph } t; \text{joining key } k$} 
        \Output{$\;\;\;\text{combined subgraph } s'$}

 \SetKwFunction{Fdissect}{dissect}
  \SetKwProg{Fn}{func}{:}{}
  \Fn{\Fdissect{$s'$, $n$}}{
\ForEach{$v$ in $s'$ \text{ in ascending order}}{
$l=$ the first $n$ vertices visited by starting from $v$ and spanning to the smallest vertex at each step\;
$r'=$ the unvisited vertices in $s'$\;
\ForEach{$v'$ in $l$ \text{ in ascending order}}{
$r=r'\cup v'$\;
\lIf{$r$ is connected}{
\Return $l,r$
}
}}}

\lIf{$s$ and $t$ have identical vertices other than $k$}{\Return $\emptyset$}
\tcp{$s'$ is a valid subgraph joined by $s$ and $t$}
$s'=s \cup t$\;
\tcp{find the smallest dissection of $s'$}
$l,r$ = \texttt{dissect}($s'$, $t.size$)\;
\tcp{if the two joining subgraphs correspond to the smallest dissection, return $s'$}
\lIf{$l==t$ \textbf{and} $r==s$}{\Return $s'$}
\lElse{\Return $\emptyset$}
 \caption{Combine two subgraphs and check for automorphism.}
 \label{alg:combine}
\end{algorithm}

Combing small subgraphs in different ways can lead to identical results. 
As we briefly mentioned in Figure~\ref{fig:reg_join}, joining subgraph `$342$' and `$375$' generates the same subgraph as joining `$352$' and `$274$'. 
These redundant subgraphs incur redundant computation, and the redundancy can accumulate over the exploration steps. 
To eliminate the redundant subgraphs, we perform an automorphism check when a subgraph is generated. 
The previous automorphism check technique for single-vertex exploration is based on the concept of the {\em canonicality} of the subgraphs~\cite{10.1145/2815400.2815410}.  
This canonicality check does not work for multi-vertex exploration because the small subgraphs are generated by a matching algorithm and may not have the canonicality property. 
We propose a \textit{smallest-vertex-first} dissection method that enables the redundancy removal for multi-vertex exploration. 

Our method is based on the following observation: for any subgraph, there is only one way to divide it into two smaller  subgraphs with both subgraphs being connected and one of them having the smallest spanning vertex indices. 
Thus, we can eliminate redundancy by producing a subgraph $s'$ only if the two joining subgraphs correspond to this unique dissection of $s'$. 
%For a subgraph obtained by joining a size-$x$ subgraph with a size-$y$ subgraph (suppose $x\le y$),  we try to divide the subgraph into a size-$x$ subgraph $l$ and a size-$y$ subgraph $r$, and find the dissection with the smallest vertices in $l$ and the smallest joining vertex in $r$. 

%Because this smallest dissection is unique for any subgraph, we ensure that only one output is generated for a subgraph. 

The automorphism check is performed each time we combine two subgraphs (i.e., in the $combine$ function in Figure~\ref{fig:join_code}). 
Algorithm~\ref{alg:combine} shows the procedure of the $combine$ function. 
For a pair of input subgraphs $s$ and $t$ ($t$ is usually a size-$3$ subgraph), we first check if there are any other identical vertices except for the joining vertex $k$. 
If yes, $s$ and $t$ cannot form a valid subgraph, and the function return an empty set. 
If no, we give the combined subgraph to a dissection procedure that divides the subgraph into two small subgraphs $l$ and $r$. 
From the vertex with the smallest index, the dissection procedure finds the smallest $n$ vertices and store them in $l$ where $n$ is the size of $t$. 
Next, the algorithm checks if the remaining vertices can constitute a connected subgraph $r$ with any of the vertices in $l$. 
If yes, the dissection procedure stops and returns $l$ and $r$. 
The algorithm returns as soon as the first dissection is found, and it will always return because of Theorem~\ref{th:1}. 
Once we have the smallest dissection $l$ and $r$, we check if they are the same as $t$ and $s$. 
If yes, the $combine$ function returns the combined subgraph; otherwise, it returns an empty set.

\textbf{Example: } The smallest-vertex-first dissection of the subgraph `$34257$' in Figure~\ref{fig:reg_graph} can be obtained by spanning from vertex $2$. 
The two adjacent vertices of $2$ are $3$ and $8$. Because $3$ is smaller, we take $3$ in the first step, and the visited set contains vertex $2$ and $3$. 
The vertices that are adjacent to the two visited vertices are $4,5,7,8$. 
Because $4$ is the smallest, we take $4$ in the next step, and we have three vertices $2,3,4$ in $l$. The unvisited vertices are $5$ and $7$. We check if any of $2,3,4$ can form a connected graph with $5,7$, and we find $3$ is the smallest vertex that connects 5 and 7. The algorithm stops and returns $l=\{2,3,4\}$ and $r=\{3,5,7\}$. 
When joining the two subgraph lists in Figure~\ref{fig:reg_join},  our system generates `$34275$' (by combining `$342$' and `$375$') instead of `$35274$' (by combing `$352$' and `$274$'). 
For the same reason, `$37485$' is not generated by combing `$374$' and `$385$' as the smallest dissection of `$37485$' is `$345$' and `$387$'. 

The worst cases complexity of the algorithm is $O(|s'|^3)$. Although it is higher than the linear complexity of the automorphism check for single-vertex exploration~\cite{10.1145/2815400.2815410, 222571}, the actual number of instructions does not increase much because $s'$ is small and the algorithm usually returns early at line 7. 

\vspace{-.5em}
\subsection{Pattern Aggregation with Index-based Quick Pattern}
\label{sec:quick_pattern}
%\vspace{-0.5em}

Next, we need to aggregate the subgraphs according to their patterns. 
This is done by computing the canonical form of each subgraph. 
The subgraphs with the same canonical form are isomorphic and will be put in the same group. 
As pointed out in \S\ref{sec:appl}, computing the canonical form is expensive, especially for large patterns. 
Previous work has used a quick pattern technique to reduce the canonical form computation. 
However, their quick patterns encode little topological information of the subgraphs, resulting in a lot of quick pattern groups of isomorphic subgraphs.  

We propose an {\em index-based quick pattern} technique that can achieve more accurate grouping of  subgraphs and reduce the overhead of canonical form computation.  
The idea is to assign an index to each pattern in a subgraph list and use the indices for computing the quick pattern of the combined subgraph. 
If a subgraph list is generated by the matching algorithm, we simply index the input patterns and store the indices with each subgraph. 
As shown in Figure~\ref{fig:reg_size3}, the size-$3$ subgraphs are obtained by matching the two size-$3$ patterns. 
We store the $pattern\_idx$ with each of its embeddings. 
When two subgraphs are combined, we construct a 4-tuple as the quick pattern for the combined subgraph. 
The first two elements in the 4-tuple are the pattern indices of the two joining subgraphs. 
The third element represents the position of the joining vertex in the two subgraphs.  
Suppose the two joining subgraphs $s_1$ and $s_2$ are of size $n_1$ and $n_2$. 
If the joining vertex is the $i$th vertex in $s_1$ and the $j$th vertex in $s_2$, then the value of the third element is ($i\times n_2+j$). 
The last element is a bitarray representing connections between the two subgraphs.  
If the $i$th vertex in $s_1$ is connected with the $j$th vertex in $s_2$, then the ($i\times n_2+j$)th bit in the bitarray is set. 

\textbf{Example: } In Figure~\ref{fig:reg_join}, the resulting subgraph `$34275$' is obtained by joining $s_1=`342$' and $s_2=`375$', and its quick pattern is $\langle 0,0,0,32\rangle$. 
The first two elements are the pattern index of `$342$' and `$375$'. The third element is 0 because the joining vertex is at position 0 in both subgraphs. The last element is $32$ because the $s_1[1]=4$ is connected with $s_2[2]=5$ in the graph and the $(1\times3+2)$th bit is set in the bitarray. 
Similarly, the quick pattern of both `$34287$' and `$35287$' is $\langle 0,0,0,128\rangle$, and the quick pattern of `$34285$' is $\langle 0,0,0,272\rangle$. 

By encoding the sub-pattern information, our quick pattern achieves more accurate grouping of the subgraphs and thus reduces the canonical form computation. 
The computation is further reduced by multi-vertex exploration as larger subgraphs contains more accurate sub-pattern information. 
To see this point, let us consider the number of possible size-4 unlabeled patterns. 
We have known that any size-4 subgraph can be obtained by joining a size-3 subgraph and an edge. 
The total number of possible 4-tuples with our index-based quick pattern is 48 ($=2\times 1\times 6\times 4$) where $2$ represents there are two types of size-3 subgraphs (i.e., triangle and wedge), $6$ is the number of possible joining positions, and $4$ is the number of possible values of the last element in the 4-tuple. 
In comparison, if we use the edge list as the quick pattern as in previous work~\cite{10.1145/2815400.2815410, 222571}, the fully-connected size-4 graph alone has 624 ($=6!-4\times 4!$) possible quick patterns where $6!$ represents all possible permutations of the six edges and $4\times 4!$ represents the permutations that do not have adjacent edges. 
This indicates that our index-based technique has much fewer possible patterns compared with the technique used in previous work, leading to fewer groups for isomorphism check. 

The quick pattern is computed after every $combine$ function in Figure~\ref{fig:join_code}. 
If the $combine$ function returns a valid subgraph, we compute its quick pattern and look for the quick pattern in a global dictionary. 
The dictionary keeps a mapping from quick patterns to their indices. 
If the quick pattern exists, we store its index with the subgraph. 
If a quick pattern is not found in the dictionary, we increase the global index number and insert a new pair of quick pattern and its index. 
In our implementation, we parallelize the for-loop that iterates over all keys in the first joining hash table. 
To avoid synchronization among threads, we store a quick pattern dictionary for each thread.

\subsection{Exploration Space Pruning}
\label{sec:pruning}

\begin{figure}[t]
    \centering
    \includegraphics[scale=0.39]{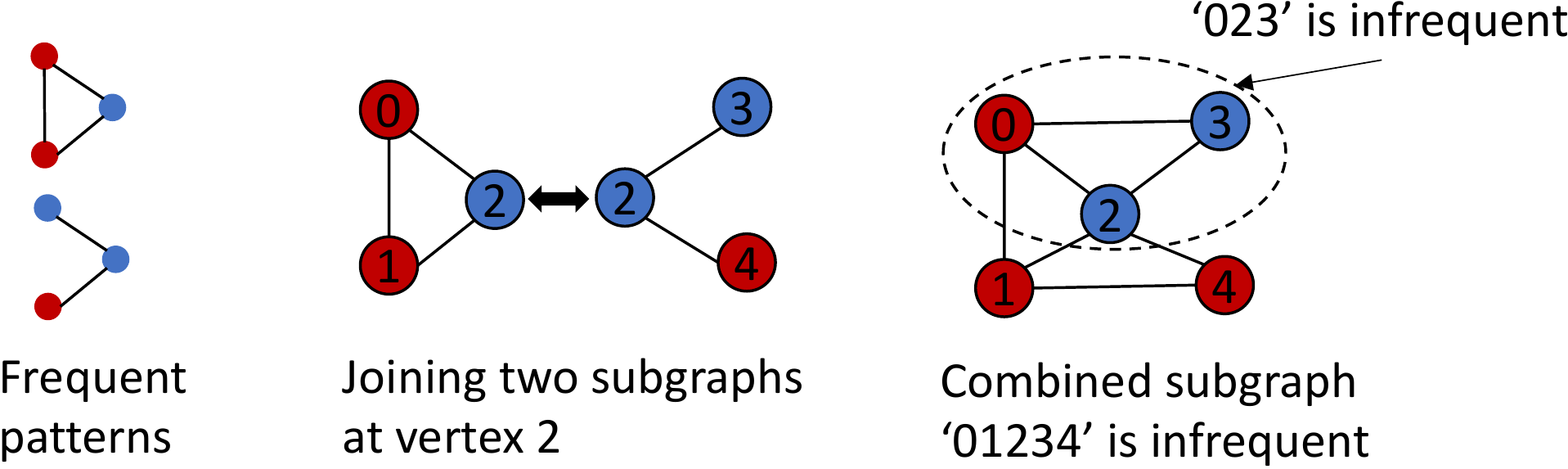}
    \vspace{-0.5em}
    \caption{An example of exploration space pruning. }
    \label{fig:pruning}
        \vspace{-0.4em}
\end{figure}

An optimization that most graph mining systems adopt for frequent subgraph mining is to filter out the subgraphs of infrequent patterns so as to reduce the subgraph exploration space~\cite{10.1145/2815400.2815410, 222571, 10.14778/3389133.3389137, 10.1145/3342195.3387548}. 
All of the existing systems achieve this optimization with breadth-first exploration. 
They either store all intermediate subgraphs (e.g., RStream~\cite{222571}, Pangolin~\cite{10.14778/3389133.3389137}) or maintain a list of frequent patterns and re-match these pattern (e.g., Peregrine~\cite{10.1145/3342195.3387548}) in each exploration step. 
The problem with the first approach is that it takes a lot memory and needs to aggregate the subgraphs in each step.  
The problem with the second approach is that it needs to perform redundant matching in each step, and it only works for support measures that can be computed without storing all the embeddings (e.g., MNI). 
If the user wants to use more accurate support measures (e.g., MIS, MVC~\cite{10.1145/3035918.3035936}), the second approach will not work. 
An advantage of two-vertex exploration is that it enables exploration space pruning without storing intermediate results or re-matching. 

Our main idea is that, instead of checking the support of the combined pattern, we check whether the vertices around the joining point form any subgraphs of smaller infrequent patterns. 
If an infrequent subgraph is found, then the combined subgraph must be infrequent and should be discarded.  
Figure~\ref{fig:pruning} shows an example of this method. 
When the system tries to join two subgraph `012' and `234' at vertex 2, if finds that there is an edge connecting vertex 0 and 3 and an edge connecting vertex 1 and 4. 
This forms two triangles `023' and `124'. While triangle `124' is frequent, triangle `023' is not, according to the list of frequent size-3 patterns. 
Due to the anti-monotone property of the support measure, a frequent pattern cannot contain infrequent subpatterns. 
Thus, the combined subgraph `01234' must be infrequent and should not be used for further exploration.

The above pruning procedure is done in the $combine$ function (line 9 in Algorithm~\ref{alg:combine}) when we check the connectivity among vertices of the two joining subgraphs. 
For any size-3 subgraph `abc' with `a' being the joining vertex and `b', `c' from different subgraphs, if the subgraph is not in the list of frequent size-3 patterns, the $combine$ function returns an empty set immediately.

%Although this pruning procedure cannot eliminate all the infrequent subgraphs, our experiments show that our system can effectively benefit from the anti-monotone property for frequent subgraph mining. 
%The advantage of our pruning technique is that it avoids the expensive storage and aggregation of intermediate subgraphs. 
%Compared with the pattern-based approach of Peregrine~\cite{10.1145/3342195.3387548}, our technique works for any support measure with the anti-monotone property.  

%% file: text/approx_join.tex
\vspace{-1em}
\section{Subgraph Sampling for Faster Exploration}
\label{sec:approx_join}
%\vspace{-0.5em}

In real applications, we may not need to find all frequent patterns, and exhaustive exploration is unnecessary~\cite{222637}. 
Thus, we propose a {\em subgraph sampling} technique to accelerate the exploration process. 
%{\systemName} provides the interface that allows the user to specify the approximation parameters easily. We also find that our approach incurs negligible errors in the results. 

\textbf{Sampling during Joining:}
The idea is to add a sampling operation each time we iterate over the joining subgraphs, i.e., before each for-loop in the dotted boxes in Figure~\ref{fig:join_code}. 
Because the MNI support measures the frequency of a pattern as the  number of distinct matching vertices, 
we sample a fixed number of iterations in each of the boxed for-loops in Figure~\ref{fig:join_code}, in order to achieve a more even distribution of  subgraphs over all vertices. 
If a loop has fewer iterations than the sampling threshold, we execute all of them; if a loop has more iterations than the threshold, we sample the iterations uniformly to the threshold number. 
This subgraph sampling during the joining phase can be considered as a generalization of the {\em neighbor sampling} technique in ASAP~\cite{222637}. 
ASAP samples a subset of the edges when it extends the matched subgraph from one vertex to its neighbors. 
We sample the neighboring size-3 subgraphs instead. 
Intuitively, our subgraph sampling is more accurate than neighbor sampling because size-3 subgraphs preserve more graph structures than edges. 
%We also find that when multi-vertex exploration is used the estimation error can be greatly reduced compared with single-vertex exploration given the same overall sampling ratio. 
%This is because the variance of the estimate grows with the number of sampling stages~\cite{cochran1953sampling} and multi-vertex exploration reduces the number of join steps. 
%Intuitively, subgraphs (of more than 2 vertices) preserve more structural information of the graph than edges.  

\textbf{Sampling during Matching: }
For very large graphs, we may not be able to store all size-3 subgraphs in memory or even on disk. 
To achieve fast mining, we can sample the subgraphs during the matching phase and only store the sampled subgraphs.  
Similar to the sampling in the joining phase, we sample a fixed number of subgraphs around each vertex in order to have subgraphs evenly distributed over all vertices. 
More specifically, we permute the vertex list at each inner loop of the nested for-loop generated by AutoMine~\cite{10.1145/3341301.3359633}. 
The execution continues to the next iteration of the outermost loop if $t$ subgraphs have been matched in the current iteration. 
This will give us $t$ subgraphs sampled from each vertex. 
We set $t$ to a number such that all the sampled subgraphs can be stored in memory. 
These sampled size-3 subgraphs are then given to the join procedure to explore larger subgraphs. 
This subgraph sampling during the matching phase can be considered as a generalization of the {\em edge sampling} technique for approximate graph processing~\cite{agarwal2013blinkdb, zou2010frequent}. 
Previous work has shown that edge sampling does not work well for graph mining tasks~\cite{222637}. 
Our subgraph sampling is much more robust than edge sampling for graph pattern mining as it preserves more structures of the graph. Our experiments also validate this point.

%% file: text/experiments.tex
\section{Experimental Results}
%\vspace{-0.5em}

This section presents our experimental setup and performance comparison with the existing graph mining systems and methods. 
%\footnote{We also try to compare with Peregrine~\cite{ng}. However, Peregrine returns wrong results and has unstable performance. 
%The author confirmed with us that Peregrine has some bugs at the time of writing this paper. }: AutoMine~\cite{10.1145/3341301.3359633}, Fractal~\cite{10.1145/3299869.3319875} and  Pangolin~\cite{10.14778/3389133.3389137}.  

\subsection{Experimental Setup}

\noindent \textbf{Platform: } 
We run all the experiments on
a workstation with an Intel Xeon W-3225 CPU containing 8 physical cores (16 logical
cores with hyper-threading), 196GB memory, and a 4TB SSD. 
We use GCC 7.3.1 for compilation with optimization level O2 enabled. 
All the systems are configured to run with 16 threads.  
We use OpenMP to parallelize the for-loop that iterates over all keys in the first joining hash table.

\begin{table}[t]
\caption{Graph datasets}
\vspace{-1em}
\centering
\footnotesize
\begin{tabular}{c|c|c|c}
Graphs       & \#vertices & \#edges & Description          \\ \hline\hline
CiteSeer (CI)~\cite{10.14778/2732286.2732289}      & 3264       & 4536    & Publication citation \\ \hline
MiCo (MI)~\cite{10.14778/2732286.2732289}          & 100K     & 1.1M & Co-authorship        \\ \hline
Orkut (OK)~\cite{orkut}         & 3.1M       & 117.2M  & Social network       \\ \hline
UK-2005 (UK)~\cite{youtube}       & 39M       & 936M      & Social network       \\ \hline
Friendster (FR)~\cite{yang2015defining}       & 65M       & 1.8B      & Social network       \\ \hline
\end{tabular}
\label{tab:datasets}
\end{table}

\noindent \textbf{Datasets: }
We test on five graphs as listed in Table.~\ref{tab:datasets}. 
These graphs are commonly used for evaluating performance of graph mining systems. 
CiteSeer and MiCo are labeled, and the other four are unlabeled. 
For the unlabeled graphs, we randomly assign $30$ labels to the vertices. 

\noindent \textbf{Settings: }
We compare our system with three state-of-the-art graph mining systems:  Peregrine~\cite{10.1145/3342195.3387548} and AutoMine~\cite{10.1145/3341301.3359633} which represent the pattern-based systems, and Pangolin~\cite{10.14778/3389133.3389137} which represents the explore-aggregate-filter systems.   
We run edge-induced FSM since it is more commonly evaluated by the existing graph mining systems~\cite{222571, 10.1145/3299869.3319875, 10.1145/3342195.3387548}. 
The original code of AutoMine only supports vertex-induced FSM (which has much less computation than edge-induced FSM) and uses number of embeddings as the support measure (which is not anti-monotone). We adapt the code to support edge-induced FSM with MNI support, and we use it to find all size-3 subgraphs for our two-vertex exploration.

%Our task is to mine frequent subgraphs of 4,5,6,7 vertices. 
For most graphs, we set the MNI support threshold $t = 0.001n$, $0.005n$, $0.01n$ and $0.05n$ where $n$ is the number of nodes in the graph. 
The reason we use  proportional thresholds is that the MNI support measures frequency as the number of distinct vertices~\cite{10.1145/3035918.3035936}. 
The threshold means that if every vertex in a pattern maps to at least $t$ different vertices in the graph, we consider the pattern frequent. 
For UK and FR, because $n$ is large,  there are few patterns that can meet threshold $0.001n$. Therefore, we test with $0.0001n$ and $0.0005n$ on UK and FR.

\subsection{Performance without Sampling}
%\vspace{-0.3em}

\begin{table}[t]
\centering

\caption{Execution times in seconds. Systems: Two-Vertex exploration (TV), 
Peregrine (PR), AutoMine (AM) and Pangolin (PG). 
`T' represents timeout after 24 hours of execution. 
`F' execution failure due to insufficient memory or disk space. 
}
\label{tab:fsm_time}
\vspace{-.8em}
\footnotesize
\begin{tabular}{c|c|c||c|c|c|c}

\multicolumn{1}{c|}{ Size}                   & { Support}  & \multicolumn{1}{c||}{ Gr.}                   & \multicolumn{1}{c|}{ TV}    & { PR}                 & {AM}                     & { PG}                 \\ \hline \hline
\multicolumn{1}{c|}{\multirow{4}{*}{4-FSM}} & 0.001 & \multicolumn{1}{c||}{\multirow{4}{*}{CI}} & \multicolumn{1}{c|}{0.99}  & 5.4                & \multirow{4}{*}{5.1}   & 5.5                \\ \cline{2-2} \cline{4-5} \cline{7-7} 
\multicolumn{1}{c|}{}                       & 0.005 & \multicolumn{1}{c||}{}                    & \multicolumn{1}{c|}{0.89}  & 4.8                &                        & 4.8                \\ \cline{2-2} \cline{4-5} \cline{7-7} 
\multicolumn{1}{c|}{}                       & 0.01  & \multicolumn{1}{c||}{}                    & \multicolumn{1}{c|}{0.81}  & 3.4                &                        & 3.7                \\ \cline{2-2} \cline{4-5} \cline{7-7} 
\multicolumn{1}{c|}{}                       & 0.05  & \multicolumn{1}{c||}{}                    & \multicolumn{1}{c|}{0.61}  & 1.1                &                        & 2.8                \\ \hline
\multicolumn{1}{c|}{\multirow{4}{*}{4-FSM}} & 0.001 & \multicolumn{1}{c||}{\multirow{4}{*}{MI}} & \multicolumn{1}{c|}{41645} & \multirow{4}{*}{F} & \multirow{4}{*}{78244} & \multirow{4}{*}{F} \\ \cline{2-2} \cline{4-4}
\multicolumn{1}{c|}{}                       & 0.005 & \multicolumn{1}{c||}{}                    & \multicolumn{1}{c|}{32763} &                    &                        &                    \\ \cline{2-2} \cline{4-4}
\multicolumn{1}{c|}{}                       & 0.01  & \multicolumn{1}{c||}{}                    & \multicolumn{1}{c|}{29698} &                    &                        &                    \\ \cline{2-2} \cline{4-4}
\multicolumn{1}{c|}{}                       & 0.05  & \multicolumn{1}{c||}{}                    & \multicolumn{1}{c|}{25682} &                    &                        &                    \\ \hline
\multirow{4}{*}{5-FSM}                       & 0.001 & \multirow{4}{*}{CI}                      & 25.2                       & \multirow{4}{*}{F} & \multirow{4}{*}{68.2}  & \multirow{4}{*}{F} \\ \cline{2-2} \cline{4-4}
                                             & 0.005 &                                          & 22.1                       &                    &                        &                    \\ \cline{2-2} \cline{4-4}
                                             & 0.01  &                                          & 21.5                       &                    &                        &                    \\ \cline{2-2} \cline{4-4}
                                             & 0.05  &                                          & 16.9                       &                    &                        &                    \\ \hline
\multirow{4}{*}{6-FSM}                       & 0.001 & \multirow{4}{*}{CI}                      & 615                        & \multirow{4}{*}{F} & \multirow{4}{*}{1924}  & \multirow{4}{*}{F} \\ \cline{2-2} \cline{4-4}
                                             & 0.005 &                                          & 597                        &                    &                        &                    \\ \cline{2-2} \cline{4-4}
                                             & 0.01  &                                          & 564                        &                    &                        &                    \\ \cline{2-2} \cline{4-4}
                                             & 0.05  &                                          & 416                        &                    &                        &                    \\ \hline
\multirow{4}{*}{7-FSM}                       & 0.001 & \multirow{4}{*}{CI}                      & 26760                      & \multirow{4}{*}{F} & \multirow{4}{*}{63362} & \multirow{4}{*}{F} \\ \cline{2-2} \cline{4-4}
                                             & 0.005 &                                          & 24644                      &                    &                        &                    \\ \cline{2-2} \cline{4-4}
                                             & 0.01  &                                          & 23697                      &                    &                        &                    \\ \cline{2-2} \cline{4-4}
                                             & 0.05  &                                          & 16257                      &                    &                        &                    \\ \hline
\end{tabular}
\vspace{-.5em}
\end{table}

Since none of the compared systems supports sampling, we first run our algorithm without  sampling to compare the performance. 
Table~\ref{tab:fsm_time} summarizes the execution time of FSM for which at least one of the compared systems can return result within 24 hours.  
The execution time of our system reported here is the time of Step 2,3,4 as described in Section~\ref{sec:overview}. 
We do not include the time for Step1 because 1) it is negligible on these two graphs (0.08 seconds on CI and 102 seconds on MI) compared with the joining time, and 2) Step1 can be considered as preprocessing. 
We find that Peregrine and Pangolin abort for most tasks.  
In fact, Peregrine paper~\cite{10.1145/3342195.3387548} only reports results of 3-FSM. Pangolin~\cite{10.14778/3389133.3389137} reports results mostly for 3-FSM. It reports 4-FSM for only one graph using large support thresholds, but it fails to give result for MI. 
For the only one testcase (4-FSM on CI) that Peregrine and Pangolin do return, our system is 1.8x to 5.6x faster. 
AutoMine is able to return results for these tasks. However, because it matches the patterns in a depth-first order, it cannot benefit from the anti-monotone property (i.e., it does not run faster for larger support thresholds). 
Our system is 1.9x to 8.4x faster than AutoMine for these tasks. 

\begin{figure*}[t]
    \centering
    \includegraphics[scale=0.45]{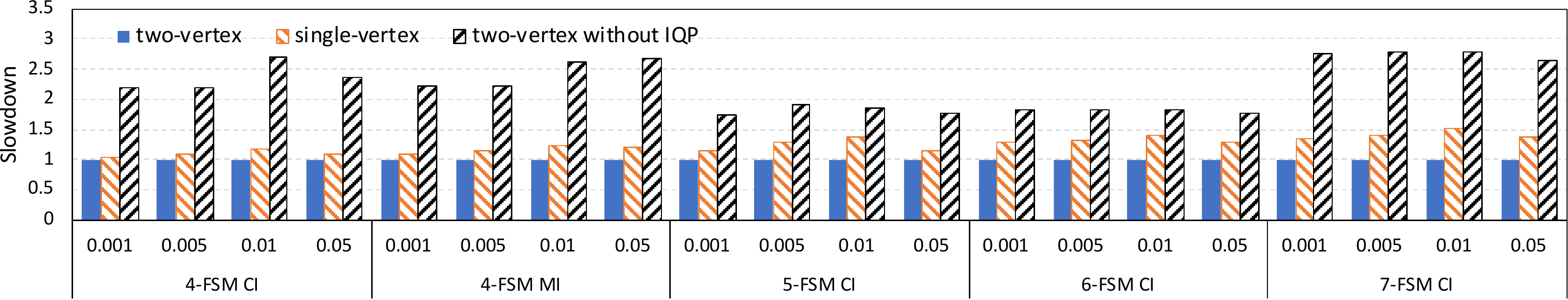}
     %   \vspace{-1em}
    \caption{Performance of two-vertex exploration, single-vertex exploration, and two-vertex exploration without our index-based quick pattern. The Execution times are normalized for each task with the execution time of two-vertex exploration in Table~\ref{tab:fsm_time}. }
        \vspace{-1em}
    \label{fig:slowdown}
\end{figure*}

\begin{figure}[t]
    \centering
    \includegraphics[scale=0.45]{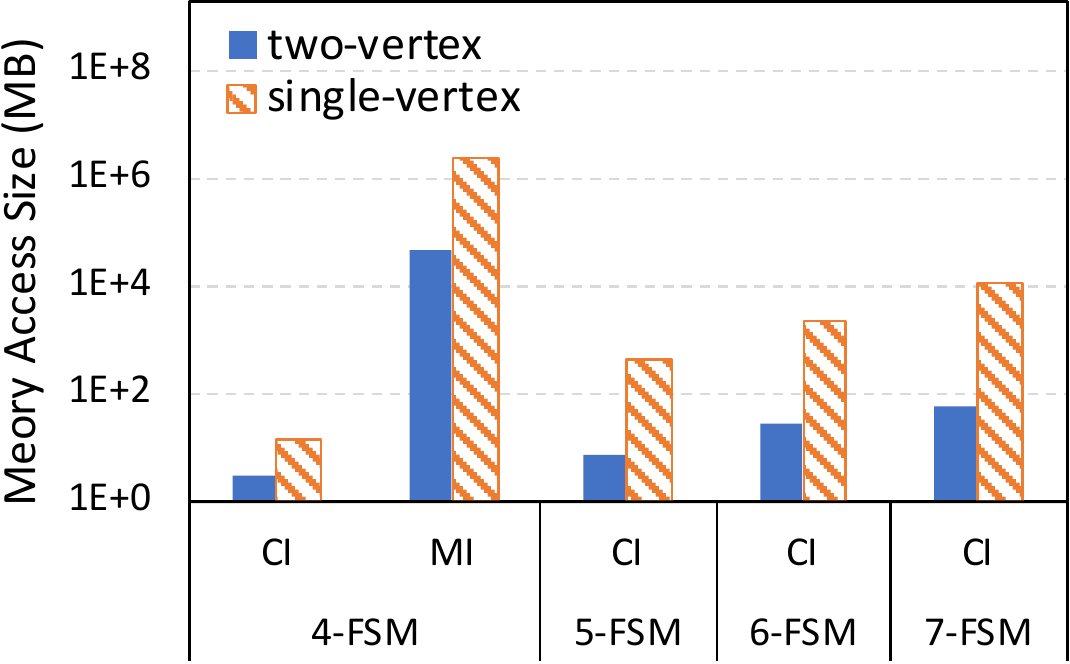}
        \vspace{-.5em}
    \caption{Total memory access size in multi-way join with two-vertex and single-vertex exploration for different FSM tasks (MNI support threshold $0.001n$). }
        \vspace{-1em}
    \label{fig:memory_load}
\end{figure}

\noindent \textbf{Advantage over Single-Vertex Exploration: }
As discussed in Section~\ref{sec:multi-join}, one advantage of two-vertex exploration over single-vertex exploration is that it reduces the memory access overhead in depth-first multi-way join.  
To show the advantage, we configure our system to run single-vertex exploration. 
The single-vertex version still uses our index-based quick patterns, but it does not support exploration space pruning since the size-3 subgraphs are not computed. 
The execution times of single-vertex exploration are shown in
Figure~\ref{fig:slowdown}. 
We can see that single-vertex exploration is 1.02x to 1.52x slower than two-vertex exploration. 
We also collect the total memory access sizes to the hash tables with two-vertex exploration and single-vertex exploration (assuming every query to the hash tables is a cache miss). 
As shown in Figure~\ref{fig:memory_load},
 two-vertex exploration reduces the memory access overhead by 5x to 189x. 
The results are collected with support threshold $0.001n$. Other support thresholds show a similar pattern.

%\rw{Do we minimize the memory access time compares with other systems? fig 6 and 7 do not show prior schemes' data?}

\begin{figure}[t]
    \centering
    \includegraphics[scale=0.45]{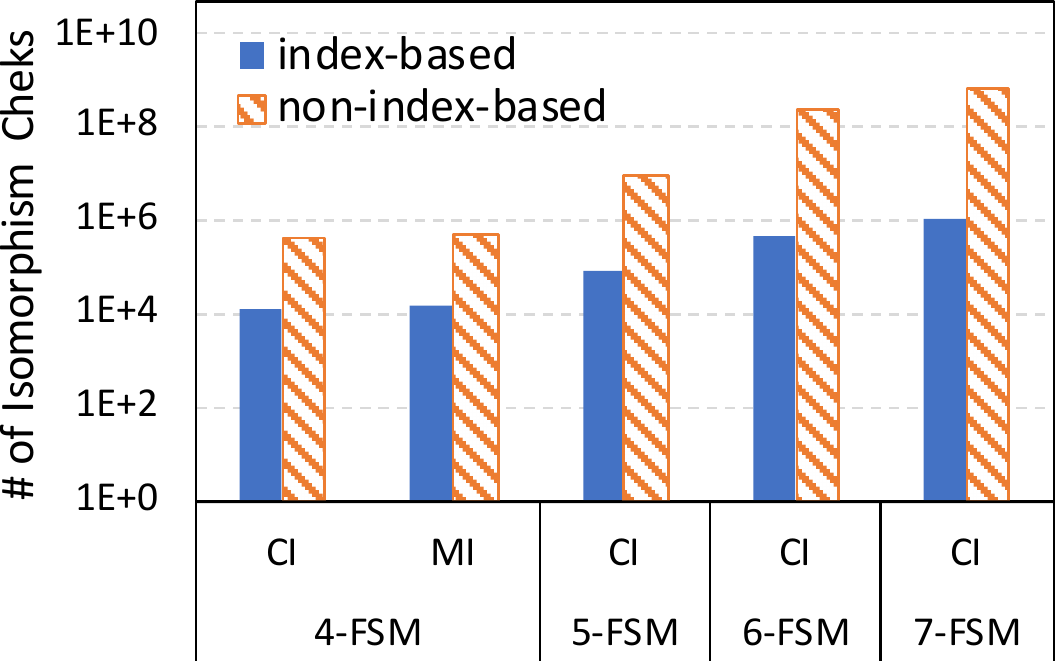}
    %\vspace{-.5em}
    \caption{Number of isomorphism checks for different FSM tasks (MNI support threshold $0.001n$) with and without our index-based quick pattern technique. }
    \label{fig:iso_checks}
\end{figure}

\noindent \textbf{Benefit of Index-Based Quick Pattern: }
To show the benefit of our index-based quick pattern technique, we disable our index-based quick pattern and use the quick pattern technique in previous work instead (i.e., a list of edges with labels of adjacent nodes). 
Figure~\ref{fig:slowdown} shows the execution times of two-vertex exploration without our index-based quick pattern. 
We can see that it leads to 1.75x to 2.78x slowdown.  
To further verify the advantage, 
we collect the number of invocations to the bliss function~\cite{JunttilaKaski:ALENEX2007} for computing the canonical forms of subgraphs. 
As shown in Figure~\ref{fig:iso_checks}, our index-based quick pattern reduces the number of isomorphism checks by 31x to 564x for different tasks, which explains the speedups.

\subsection{Performance with Sampling}
%\vspace{-0.5em}

Next, we evaluate the effectiveness of our sampling methods. 
Since all the size-3 subgraphs of CI and MI can be stored in memory, we only perform sampling during the joining phase for these two graphs.

\begin{figure}[t]
    \centering
    \includegraphics[scale=0.45]{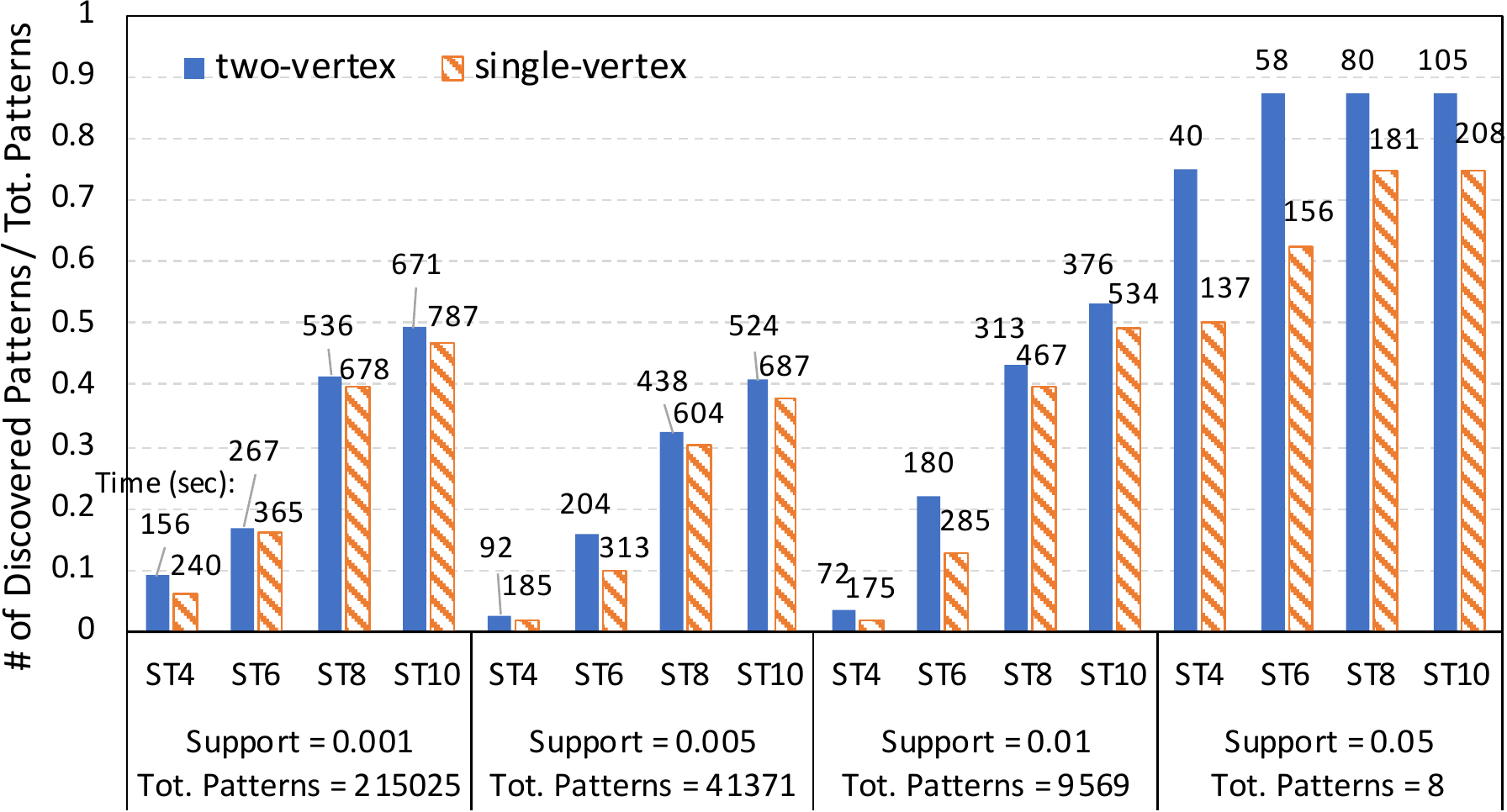}
    \vspace{-.5em}
    \caption{Number of discovered size-4 frequent patterns on MI graph with different support thresholds and different sampling thresholds. For single-vertex exploration, ST$x$ means in every join step only $x$ edges are sampled in each neighbor list. For two-vertex exploration, it means that $x$ edges  and $x^2$ size-3 subgraphs are sampled in each key group when we join the edge list with the size-3 subgraphs. }
    \label{fig:fsm_s_mi4}
\end{figure}

\begin{figure}[t]
    \centering
    \includegraphics[scale=0.45]{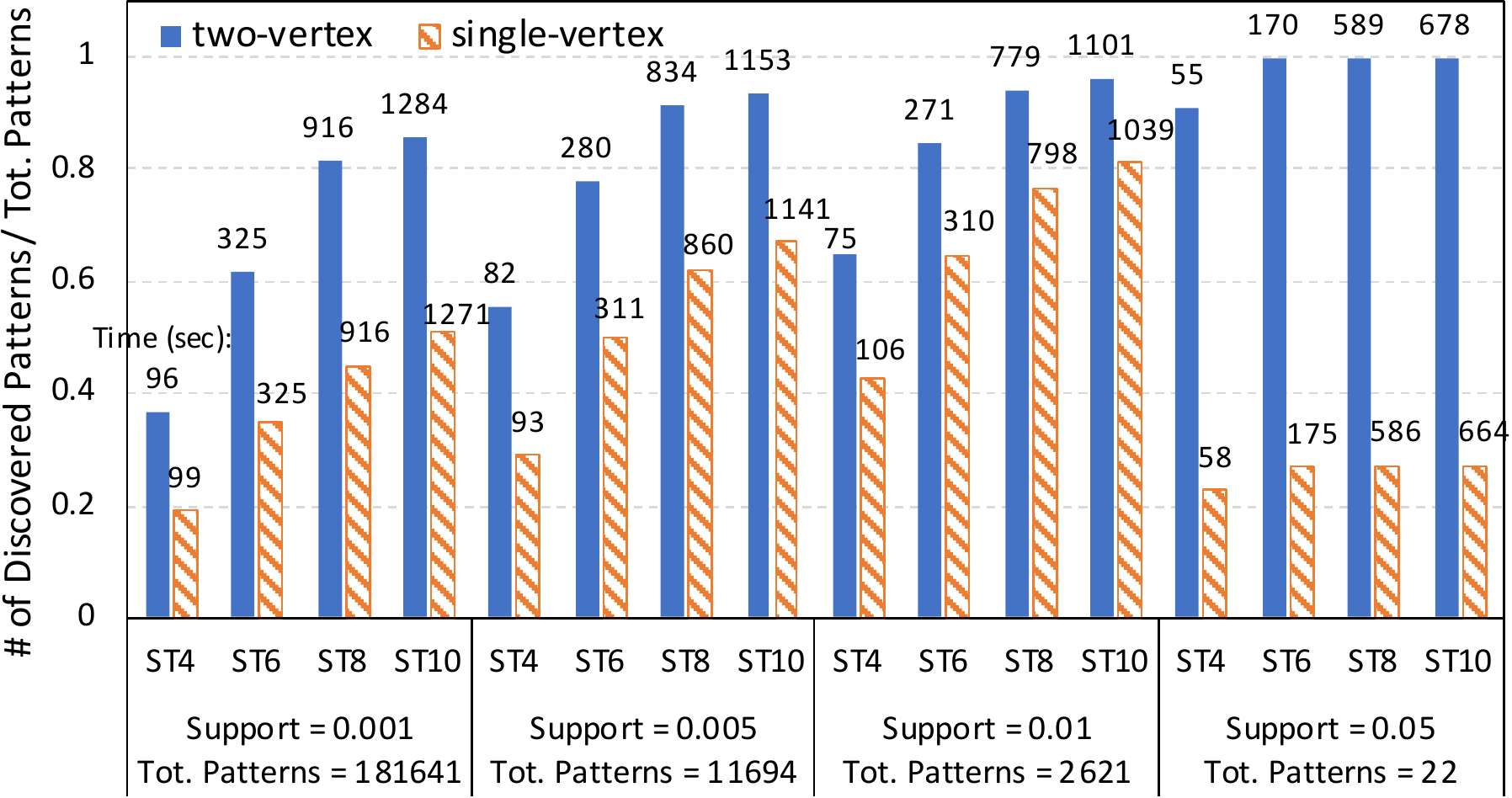}
    \vspace{-.5em}
    \caption{Number of discovered size-7 frequent patterns on CI graph with different support thresholds and different sampling thresholds. For single-vertex exploration, ST$x$ means in every join step only $x$ edges are sampled in each neighbor list. For two-vertex exploration, it means that in every join step only $x^2$ size-3 subgraphs are sampled in each key group. }
    \label{fig:fsm_s_ci7}
\end{figure}

Figure~\ref{fig:fsm_s_mi4} shows the number of size-4 frequent patterns we can find on MI graph with different support thresholds and different sampling thresholds. 
The execution time of different runs are labeled on top of the bars. 
When the support threshold is set to $0.001n$, there are 215025 frequent patterns in total, and to discover all these patterns precisely our system needs to run for 41645 seconds (as shown in Table~\ref{tab:fsm_time}).  
If we sample $10$ edges and $100$ size-3 subgraphs in each key group when we join the edge list and the size-3 subgraphs (ST10 in the figure), our two-vertex exploration returns 49\% of the frequent patterns in 671 seconds. The execution time is reduced by 62x.  
When the support threshold is set to $0.005n$, we can find 41\% of the total frequent patterns within 524 seconds, which is 1/63 of the total execution time. 
When the support threshold is set to $0.05n$, there are only 8 frequent patterns, and our two-vertex exploration with ST6 sampling can find 7 of them in 58 seconds, which leads to a 443x speedup compared with the accurate execution. 
The figure also shows that the larger sampling thresholds we use the more frequent patterns we can find. 

Figure~\ref{fig:fsm_s_ci7} shows the number of size-7 frequent patterns found on CI graph with different support thresholds and different sampling thresholds. 
When the support threshold is $0.001n$, our two-vertex exploration can find 86\% of the frequent patterns in 1284 seconds with ST10 sampling. 
Compared with the time of accurate execution in Table~\ref{tab:fsm_time}, sampling achieves a 21x speedup. 
When the support threshold is set to $0.005n$ and $0.01n$, there are fewer frequent patterns, and our two-vertex exploration with ST10 sampling can find more than 90\% of the frequent patterns with less than 1/21 of the total execution time. 
When the support threshold is $0.05n$, there are only 22 frequent patterns, and our two-vertex exploration with ST6 sampling can find all of them within 170 seconds, which is 1/96 of the accurate execution time.

\begin{figure}
\centering
\captionof{table}{Results of 9-FSM on CI graph with sampling threshold 4 for two-vertex exploration (TV) and sampling threshold 2 for single-vertex exploration (SV). }
\vspace{-.6em}
\subfloat[Number of discovered patterns]{
\scriptsize
\begin{tabular}{c||c|c|c|c}
Support & 0.001    & 0.005 & 0.01 & 0.05     \\ \hline\hline
TV &63941 & 6050 & 1770 & 16 \\ \hline
SV & 402 &  0 & 0 & 0 \\ \hline
\end{tabular}
}\hfil
\subfloat[Execution time (sec)]{
\scriptsize
\begin{tabular}{c|c|c|c}
 0.001    & 0.005 & 0.01 & 0.05     \\ \hline\hline
 73 & 71 & 54 & 45 \\ \hline
 75 &  75 & 63 & 46 \\ \hline
\end{tabular}
}
\label{tab:fsm9}
\end{figure}

\noindent \textbf{Advantage over Single-Vertex Exploration: } 
As discussed in Section~\ref{sec:approx_join}, another advantage of two-vertex exploration over single-vertex exploration is that it leads to more accurate sampling.  
To verify this, we configure our system to run sampled single-vertex exploration. 
Since single-vertex exploration needs twice join steps as two-vertex exploration, we set its sampling threshold to the square root of the threshold for two-vertex exploration in order to achieve a similar size of overall exploration space. 
As shown in Figure~\ref{fig:fsm_s_mi4} and~\ref{fig:fsm_s_ci7}, if we do not include the matching time, two-vertex exploration has a slightly shorter execution time than single-vertex exploration when they use the corresponding sampling thresholds. 
Even if we add the time for matching size-3 subgraphs (102 seconds on MI and 0.08 seconds on CI), the total execution time is close to that of single-vertex exploration. 
For 4-FSM on MI, two-vertex exploration finds 6\% to 57\% more frequent patterns than single-vertex exploration. 
For 7-FSM on CI,  the number of frequent patterns found by two-vertex exploration is 1.18x to 4x that of single-vertex exploration. 
Table~\ref{tab:fsm9} shows the number of size-9 frequent patterns found on CI graph with different support thresholds. 
We can see that with a similar execution time two-vertex exploration discovers much more frequent patterns than single-vertex exploration. 
It is worth noting that none of the previous systems can return results for 9-FSM even on a small graph like CI. 
AutoMine cannot even enumerate all the size-9 unlabeled patterns in 24 hours. 

\begin{figure}
\centering
\captionof{table}{Results of 5-FSM on OK graph with support threshold $0.001n$ and different sampling thresholds. `M. ST' stands for Matching Sampling Threshold, `M. Time' stands for Matching Time, `J. ST' stands for Joining Sampling Threshold, `J. Time' stands for Joining Time.}
\vspace{-.6em}
\subfloat[Two-vertex exploration]{
\scriptsize
\begin{tabular}{c|c|c|c|c}
M. ST               & M. Time (sec)             & J. ST & J. Time (sec) & \# of Patterns \\ \hline\hline
\multirow{2}{*}{2}  & \multirow{2}{*}{119} & 4     & 157     & 18             \\ \cline{3-5} 
                    &                      & 16    & 168     & 18             \\ \hline
\multirow{2}{*}{4}  & \multirow{2}{*}{119} & 4     & 207     & 20             \\ \cline{3-5} 
                    &                      & 16    & 216     & 20             \\ \hline
\multirow{2}{*}{8}  & \multirow{2}{*}{120} & 4     & 268     & 20             \\ \cline{3-5} 
                    &                      & 16    & 269     & 20             \\ \hline
\multirow{2}{*}{16} & \multirow{2}{*}{122} & 4     & 299     & 20             \\ \cline{3-5} 
                    &                      & 16    & 318     & 20             \\ \hline
\end{tabular}
\label{tab:fsm5_ok_two}
}

\subfloat[Single-vertex exploration]{
\scriptsize
\begin{tabular}{c|c|c}

J. ST & J. Time & \# of Patterns \\ \hline\hline
1     &     1184    &       15         \\ \hline
2     &   6780      &      18          \\ \hline
\end{tabular}
\label{tab:fsm5_ok_single}
}
\end{figure}

\begin{figure}
\centering
\captionof{table}{Results of 5-FSM on UK and FR graph with different sampling thresholds and support thresholds. `M. ST' stands for Matching Sampling Threshold, `M. Time' stands for Matching Time, `J. ST' stands for Joining Sampling Threshold, `J. Time' stands for Joining Time.}
\label{tab:uk_fr}
\vspace{-.6em}
\subfloat[Two-vertex exploration on UK]{
\scriptsize
\begin{tabular}{c|c|c|c|c|c}

Support                 & M. ST & M. Time (sec)        & J. ST & J. Time (sec) & \# of Patterns \\ \hline \hline
\multirow{4}{*}{0.0001} & \multirow{2}{*}{2}          & \multirow{2}{*}{2254} & 4                         & 140          & 76             \\ \cline{4-6} 
                        &                             &                       & 16                        & 288         & 143            \\ \cline{2-6} 
                        & \multirow{2}{*}{4}          & \multirow{2}{*}{2530} & 4                         & 185          & 105             \\ \cline{4-6} 
                        &                             &                       & 16                        & 505         & 188            \\ \hline
\multirow{4}{*}{0.0005} & \multirow{2}{*}{2}          & \multirow{2}{*}{2254} & 4                         & 105          & 1              \\ \cline{4-6} 
                        &                             &                       & 16                        & 235         & 3              \\ \cline{2-6} 
                        & \multirow{2}{*}{4}          & \multirow{2}{*}{2530} & 4                         & 138          & 8              \\ \cline{4-6} 
                        &                             &                       & 16                        & 421         & 14             \\ \hline
\end{tabular}
\label{tab:fsm5_uk}}

\subfloat[Two-vertex exploration on FR]{
\scriptsize
\begin{tabular}{c|c|c|c|c|c}

Support                 & M. ST              & M. Time (sec)              & J. ST & J. Time (sec) & \# of Patterns \\ \hline\hline
\multirow{4}{*}{0.0001} & \multirow{2}{*}{2} & \multirow{2}{*}{6010}  & 4     & 499     & 1              \\ \cline{4-6} 
                        &                    &                        & 16    & 7620   & 1              \\ \cline{2-6} 
                        & \multirow{2}{*}{4} & \multirow{2}{*}{6657} & 4     & 547     & 4              \\ \cline{4-6} 
                        &                    &                        & 16    & 9955   & 8              \\ \hline
\multirow{4}{*}{0.0005} & \multirow{2}{*}{2} & \multirow{2}{*}{6010}  & 4     & 315     & 1              \\ \cline{4-6} 
                        &                    &                        & 16    & 4637    & 1              \\ \cline{2-6} 
                        & \multirow{2}{*}{4} & \multirow{2}{*}{6657} & 4     & 377     & 1              \\ \cline{4-6} 
                        &                    &                        & 16    & 5417    & 1              \\ \hline
\end{tabular}
\label{tab:fsm5_fr}}
\end{figure}

\noindent \textbf{Results on Large Graphs: }
Since the size-3 subgraphs of OK, UK and FR cannot be entirely stored in memory, we perform sampling during the matching phase and only store the sampled size-3 subgraphs. 
Table~\ref{tab:fsm5_ok_two} shows the number of size-5 frequent patterns with support larger than $0.001n$ found on OK graph. 
In the table, a matching sampling threshold $x$ means that $x$ subgraphs are sampled from each vertex during the matching phase. 
A larger matching sampling threshold results in longer matching time, although they are not proportional -- the matching time is mainly determined by the number of subgraph groups that need isomorphism checks. 
We find that the number of discovered patterns does not increase much with larger sampling thresholds. 
This is because $0.001n$ is a relatively large support threshold for this graph, and there are not many frequent patterns. 
To show the advantage of two-vertex exploration, we run single-vertex exploration for the same task.  
Since single-vertex exploration needs not match the size-3 subgraphs, we only perform sampling during the joining phase with threshold 1 and 2. 
The results are shown in Table~\ref{tab:fsm5_ok_single}. 
We can see that single-vertex exploration takes a longer time and finds fewer patterns. 

The results of 5-FSM on UK and FR graph are shown in Table~\ref{tab:uk_fr}.
%Table~\ref{tab:fsm5_uk} shows the number of size-5 frequent patterns found on UK graph with different sampling thresholds and support thresholds. 
We can see that matching takes a large proportion of the total execution time. 
This is because there are a lot of size-3 subgraphs and patterns in these two large graphs. 
However, if we consider matching as preprocessing and store the sampled size-3 subgraphs in memory, the joining procedure is fast. 
As shown in Table~\ref{tab:fsm5_uk}, our system can find frequent patterns on UK within a few minutes, and more patterns can be found by using larger joining sampling thresholds. 
Table~\ref{tab:fsm5_fr} shows the results of 5-FSM on FR graph. 
Again, the matching procedure is expensive. 
Once the size-3 subgraph are sampled, we can find size-5 frequent patterns in a relatively short time with sampled join. 
For comparison, we also run single-vertex exploration with sampling threshold 2 on these two graphs. It cannot finish execution within 24 hours, so we print out the found patterns after 24 hours of execution. For UK, it  returns 5 frequent patterns when support threshold is set to $0.0001n$, and 0 frequent pattern when support threshold is $0.0005n$. For FR, single-vertex exploration  with sampling threshold 2 cannot return any frequent pattern within 24 hours.

%% file: text/relatedwork.tex
\section{Related Work}
%\vspace{-0.5em}
This section summarizes the graph pattern mining systems that are most related to our work.  

\noindent
\textbf{Exploration-based Systems: }
Arabesque~\cite{10.1145/2815400.2815410} is a distributed graph pattern mining system. It enumerates all possible embeddings in multiple rounds and uses a filter-process model to generate the results. It first propose the quick pattern technique for reducing isomorphism checks.  
RStream \cite{222571} is the first single-machine, out-of-core graph mining system. It supports a rich programming model that exposes relational algebra for developers to express various mining tasks and a runtime engine that can efficiently compute the relational operations. 
Pangolin \cite{10.14778/3389133.3389137} also targets  single-machine but provides GPU programming interface for acceleration. 
DistGraph~\cite{talukder2016distributed}, ScaleMine~\cite{abdelhamid2016scalemine} and G-miner~\cite{chen2018g} are all distributed graph mining systems that adopt breadth-first exploration. 
DistGraph focuses on reducing the communication of distributed computing when each node can only have a portion of the graph. 
ScaleMine proposes a two-phase mining approach to achieve good load balance and reduce communication in distributed computing. 
G-miner proposes a block-based graph partitioning technique and uses work stealing to achieve good load balance. 
Because these systems use breadth-first exploration and need to store all intermediate results, they are not able to mine large patterns on large graphs. 
Fractal~\cite{10.1145/3299869.3319875} is also exploration-based, but it supports depth-first exploration to reduce the memory consumption. 
All of the existing systems adopt single-vertex exploration. 
Our system is the first to adopt multi-vertex exploration for mining larger pattern in graphs. 

%Fractal \cite{10.1145/3299869.3319875} adopts depth-first multi-way join to avoid storing of intermediate data 

%is the distributed system we use to compare with {\systemName} and it focuses on load balancing optimization.
%Another approximation-based system, ASAP \cite{222637}, use edge-based sampling to reduce the subgraph patterns, but it can only generate approximate results.

% While distributed mining systems can leverage a huge amount of cluster resources to accelerate the mining process, they are difficult to debug and may limit the users for those who can access the expensive clusters. Therefore, single-machine graph mining systems are out for their ease-of-use and comparably high-performance.
%On the other hand, single-machine GPMI systems show their ease-of-use and comparably high-performance.

\noindent
\textbf{Pattern-based Systems: }
AutoMine~\cite{10.1145/3341301.3359633} is a single-machine graph mining system that features compiler-based optimizations. 
Their main idea is to enumerate all the unlabeled patterns of a particular size and match them one-by-one on a graph. 
Because the patterns are given, AutoMine is able to search an optimal matching strategy and combine the matching procedures of multiple patterns. 
Because of its depth-first matching order, AutoMine is hard to benefit from the anti-monotone property of FSM. 
Also, when the pattern size is more than 7, enumerating the patterns becomes difficult. 
Peregrine~\cite{10.1145/3342195.3387548} is another pattern-based system. 
Instead of enumerating all the patterns before matching, it discovers patterns based on the subgraphs it has explored and maintains a list of the patterns. 
The main issue with Peregrine is that it needs to rematch the frequent patterns in each step, which leads to redundant computation.   
DwavesGraph~\cite{chen2020dwarvesgraph} is a recently proposed pattern-based graph  mining system. It is based on the idea that the task of matching a large pattern can be divided into smaller tasks of matching the subpatterns.  
Similar to AutoMine, it needs to know all the unlabeled patterns in advance. Thus, it cannot discover large patterns. 
In fact, DwavesGraph paper only reports result for 3-FSM.

%Previous approximate GPMI and GPMA systems \cite{222637, 8411069, pavan2013counting} mainly utilize edge sampling to reduce computation overhead. For example, ASAP \cite{222637} directly samples the edges during the mining with a neighborhood sampling scheme. 
%\rw{check my writing.}
\noindent
\textbf{Approximate Pattern Mining: }
Sampling has been proposed by earlier works~\cite{al2009output, saha2015fs3} to accelerate FSM in a database of graphs. In this setting, a pattern is considered frequent if it exists in more than a certain amount of graphs. 
The main idea of these works is to perform random walk in the space of all patterns. Every time it walks from one pattern to another, it calculates a probability distribution of all candidate patterns. 
By carefully setting the sampling probability at each step, they ensure that patterns of higher supports are more likely to be sampled~\cite{al2009output}. 
More recent works consider FSM on a single graph since it is more commonly used in real applications and is more general (a list of graphs can be considered as a single graph with disconnected components)~\cite{10.14778/2732286.2732289}.  
Sampling has also proposed to accelerate pattern-based graph mining in this setting~\cite{222637, 8411069, pavan2013counting}. 
The main idea is to sample edges in the graph based on the given patterns and estimate the actual results with the sampled results.  
These methods need to know the patterns in advance. It is not obvious how they can be applied to the exploration-based systems. 
We fill this gap and show that FSM can be accelerated by sampling the subgraphs in each key group of the join operation. 

%% file: text/conclusions.tex
\section{Conclusion}
%\vspace{-0.5em}

In this work, we propose a novel two-vertex exploration method to accelerate frequent subgraph mining.  
Based on two-vertex exploration, we further improve the performance through an index-based quick pattern technique and subgraph sampling.  The experiments show that our method outperforms other state-of-the-art graph mining systems for FSM on various input graphs and pattern sizes. 